\documentclass[10pt,twocolumn]{article}

\usepackage[letterpaper,
  textwidth=6.75in,
  textheight=9.25in,
  centering]{geometry}

\usepackage[numbers,square]{natbib}
\usepackage{amsmath}
\usepackage{amssymb}
\usepackage{amsthm}
\usepackage{algorithm}
\usepackage{algorithmic}
\usepackage{xcolor}
\usepackage{comment}
\usepackage{microtype}
\usepackage[hidelinks]{hyperref}

\newtheorem{theorem}{Theorem}

\newcommand{\dd}{\mathop{}\!\mathrm{d}}

\usepackage{graphicx}

\title{Modular Bayesian Inference for Partially Misspecified Mixtures
with Uncertain Class Labels}

\author{%
Maximilian Autenrieth\textsuperscript{1,2}\thanks{Corresponding authors. \\
E-mail: \href{mailto:ma2216@cam.ac.uk}{ma2216@cam.ac.uk}; \href{mailto:kkarchev@icc.ub.edu}{kkarchev@icc.ub.edu}},
Konstantin Karchev\textsuperscript{3,4}\footnotemark[1],
David A. van Dyk\textsuperscript{2}, 
Kaisey Mandel\textsuperscript{1}, \\
and Roberto Trotta\textsuperscript{2,4}
\\[0.6em]
\small \textsuperscript{1}University of Cambridge \quad
\textsuperscript{2}Imperial College London \quad
\textsuperscript{3}University of Barcelona \quad
\small \textsuperscript{4}SISSA 
}
\date{}

\begin{document}

\maketitle

%

%

\begin{abstract}
Scientific inference often involves observations from multiple latent populations whose class memberships are uncertain and must be inferred from available data. While full Bayesian mixture modelling provides a principled approach to joint uncertainty quantification in this setting, often the model is misspecified or less trustworthy for populations that are not the target of the inference, which can distort posterior inference for the parameters of interest. We propose a modular Bayesian framework that safeguards against such biases. We decouple classification from downstream inference using fixed class probabilities, for example from a classifier trained on a subset of labelled data. We then protect inference on the parameters of interest by preventing the likelihood contributions of less trustworthy populations from influencing them. 
A key benefit of this principled modularization is that it greatly simplifies inference, reducing a fully joint mixture analysis to sampling label allocations, fitting the protected downstream model for each sample, and averaging the resulting posteriors. We establish calibration conditions on these class probabilities under which the proposed procedure preserves the inferential target of the analysis with known labels, 
and demonstrate the method in synthetic experiments and in cosmological analysis of photometric supernova samples. 
\end{abstract} 

\section{Introduction}
In many applications, populations of scientific interest are embedded in a larger mixture of latent classes, some of which may be of lesser interest and/or less reliably modelled. In these settings, class memberships are often uncertain and must be inferred from the data \citep{mclachlan2019finite,fruhwirth2019handbook}. Fully Bayesian mixture models provide a principled framework for joint uncertainty quantification across class assignments and model parameters, but misspecification of the less well understood population(s) can propagate through the mixture, altering the inferred class probabilities and distorting inference of the parameters of scientific interest.

\paragraph{Scientific motivation.}
A prominent example is type Ia supernova (SN Ia) cosmology, where an empirical standardization procedure is used to derive precise distances and constrain the expansion history of the Universe \citep{riess1998observational,perlmutter1999measurements}. This relies on the purity of the sample, which has hitherto been achieved through costly spectroscopic follow-up for type confirmation. However, current and upcoming surveys will deliver vast samples of astronomical transients observed only through broadband photometry with only scarce (\({\sim}1\,\%\)) follow-up \citep{vincenzi2023dark}, so future cosmological analyses must propagate the uncertainty of typing purely from photometry. Moreover, the contaminating non-Ia population comprises heterogeneous transient classes whose rates and brightness evolution are hard to model reliably \citep{vincenzi2023dark,moller2022dark}. Existing analyses commonly combine type ``probabilities'' from machine-learning models \citep[e.g.][]{Moller_2020,Qu_2021} with a simple analytic likelihood for the non-Ia objects \citep{kunz2007bayesian,vincenzi2023dark,abbott2024dark}, so misspecification can propagate into cosmological inference. Recent simulation-based inference (SBI) approaches \citep{karchev2023sicret,karchev2026cigars,boyd2026flowsn} add a further complication, as they require numerous realistic forward simulations of the non-Ia population, which can be extremely demanding.

The underlying problem is not specific to SN Ia cosmology. Stars observed in the same region of the sky may either belong to a stellar cluster of interest or to the surrounding field, with uncertain membership that must be propagated into inference on cluster properties of interest
\citep{jackson2020gaia,jackson2022gaia}.
The surrounding field population is typically much more heterogeneous and considerably harder to model reliably. 
Similar problems arise in galaxy-cluster cosmology \citep{wojtak2018galaxy,esteves2025copacabana}.  
In gravitational-wave astronomy, candidate events may be astrophysical signals or detector artefacts that are difficult to characterize, with scientific interest focused on the astrophysical population \citep{heinzel2023inferring}.   
In ecology, uncertain identification of similar species can propagate into inference on species-specific abundance and population dynamics
\citep{chambert2016estimating}. 
In seismology, uncertain classification of earthquakes into different tectonic regimes can affect inference on aftershock and seismicity parameters \citep{schneider2025bayesian}. 
In single-cell genomics, uncertain cell-type assignments can affect cell-type-specific inference, particularly when some populations are less well characterized \citep{zhang2019probabilistic,xu2021probabilistic,theunissen2024uncertainty}.
These settings motivate methods that propagate classification uncertainty while limiting sensitivity to less reliable component models.

\paragraph{Related literature:}
A range of Bayesian approaches have been developed to reduce sensitivity to model misspecification \citep{nott2023bayesian}. Generalized Bayesian inference replaces the usual log-likelihood contribution by a loss chosen for the inferential target, allowing  
parameters of interest to be updated without specifying and trusting a complete generative model 
\citep{bissiri2016general}. 
Restricted-likelihood methods retain a probabilistic model but update the prior using the sampling distribution of carefully chosen summary statistics, such as robust estimators, thereby reducing sensitivity to aspects of the data that the model describes poorly \citep{lewis2021bayesian}. Robust mixture models instead modify the component distributions themselves, e.g., by replacing Gaussian components with heavy-tailed Student-\(t\) distributions so that atypical observations have less
influence on the inferred mixture structure \citep{svensen2005robust}. 

When misspecification is localized to particular parts of a joint model, modular Bayesian inference provides an alternative to modifying the model globally by restricting the flow of information between model components \citep{bayarri2009modular,liu2025general}. The basic motivation is to keep a ``good'' or \emph{trusted} module separate from a \emph{suspect} module so that the latter does not contaminate inference in the former \citep{bayarri2009modular}. However, the notion of a ``module'' is not associated with a single universally adopted formal definition \citep{nott2023bayesian}. \citet{liu2025general} provide a graph-based formalization in which modules are constructed from a partition of the observable variables and the graphical structure of the joint model. In their framework, a self-contained Bayesian module includes the variables required for a well-defined Bayesian analysis for a specified set of observables without requiring information external to the module. Modules may overlap through shared parameters or latent quantities.

A principal form of modular inference is cut inference, which protects a trusted module by removing feedback from a suspect module. Information and uncertainty inferred in the trusted module may still be propagated forward when analysing the suspect module, while information from the suspect module is prevented from updating the protected quantities
\citep{plummer2015cuts,jacob2017better,liu2025general}. 
Modular Bayesian inference has received considerable recent methodological attention. Semi-modular methods allow feedback to be attenuated rather than fully removed \citep{carmona2020semi}, while variational approaches provide alternative computational strategies for cut inference \citep{yu2023variational}. 
Recent work has also extended modular inference to likelihood-free and generalized-Bayesian settings \citep{chakraborty2023modularized,frazier2025cutting}, and studied robustness, uncertainty quantification, and efficiency under misspecification \citep{smith2025cutting,frazier2026posterior}.

\paragraph{Contributions.} 
Our work extends the modular Bayesian framework to mixture models with uncertain class memberships and potentially misspecified component models. We introduce two complementary interventions. Intervention I replaces the model-based label posterior by a fixed external allocation distribution, such as probabilistic classifier outputs, preventing feedback from the downstream model into classification while propagating classification uncertainty. Intervention II removes suspect-component likelihood contributions from the downstream analysis, protecting parameters of interest from misspecification in those components. Together, they give a practical inference procedure that samples label allocations from the fixed allocation distribution, performs standard Bayesian inference with the protected downstream model conditional on each allocation, and averages the resulting posteriors.

We establish theoretical guarantees for the resulting modular procedure. We show that, under an explicit calibration condition on the adopted class probabilities, the classifier-based posterior has the same asymptotic target as the complete-label posterior. When the retained components are correctly specified, the protected analysis recovers the true parameters of interest.

We further study how the information used for classification affects downstream inference, illustrating that additional observables that improve classification can also increase the posterior \emph{efficiency}. However, we demonstrate that predictive classification performance alone does not ensure \emph{reliable} downstream inference. Finally, we demonstrate the framework and theoretical findings in a controlled synthetic example and apply the methodology in a scientific case study of photometric supernova cosmology.

\section{Method}

\subsection{Notation and problem setup}\label{sec:notation_and_fullyBayes}

We consider observations arising from \(K\) latent populations. For 
\(i=1,\ldots,n\), let
\(\pmb D_i=(\pmb X_i,\pmb Y_i)\) denote the observed data and 
\(L_i\in\{1,\ldots,K\}\) the corresponding component label. 
We distinguish between \(\pmb Y_i\), 
the data used for downstream inference on the scientific parameters of interest, 
and \(\pmb X_i\), additional information about class (component) membership. 
In general, \(\pmb X_i\) and \(\pmb Y_i\) may be dependent within components, and inference on class membership may use one or both. 
Labels may be observed for a subset of objects and unknown for the remainder. 
Bold symbols denote collections, e.g.,
\(\pmb D=(\pmb D_1,\ldots,\pmb D_n)\).

A fully Bayesian analysis specifies a joint model for the labels and all observed data. In a standard mixture formulation, \(L_i\mid\pmb\pi\sim\operatorname{Categorical}(\pmb\pi)\), with \(\pmb\pi\) denoting the a-priori class proportions, and
\[
    p(\pmb D_i,L_i=k\mid\pmb\theta)
    =
    \pi_k\,
    p_k(\pmb Y_i\mid\pmb\psi,\pmb\gamma_k)\,
    h_k(\pmb X_i\mid\pmb Y_i,\pmb\phi_k),
\]
where \(p_k\) is the downstream sampling distribution in component \(k\), and \(h_k\) models the additional information about class membership carried by
\(\pmb X_i\).
%
We assume the observations are conditionally independent across \(i\) given the latent labels and model parameters. 
Here \(\pmb\psi\) denotes parameters shared across components, and \(\pmb\phi=(\pmb\phi_1,\ldots,\pmb\phi_K)\) and \(\pmb\gamma=(\pmb\gamma_1,\ldots,\pmb\gamma_K)\) are the component-specific downstream parameters that govern the generative mixture models for \(\pmb X\) and \(\pmb Y\) respectively. We write \(\pmb\theta=(\pmb\psi,\pmb\gamma,\pmb\pi,\pmb\phi)\) for the full parameter vector. The full joint model is represented graphically in the left panel of Fig.~\ref{fig:graph}. 

Conditional on the labels, \(\pmb X\) contains no further information about
\((\pmb\psi,\pmb\gamma)\), since its distribution does not depend on them, so that
\(p(\pmb\psi,\pmb\gamma\mid\pmb D,\pmb\ell)
= p(\pmb\psi,\pmb\gamma\mid\pmb Y,\pmb\ell)\). 
With independent priors for \(\pmb\psi\), 
\(\pmb\gamma\), \(\pmb\pi\), and \(\pmb\phi\), the downstream posterior
can be written by marginalising over the unknown labels as
\begin{equation}\label{eq:full-posterior-decomposition}
    p(\pmb\psi,\pmb\gamma\mid\pmb D)
    = \sum_{\pmb\ell}
    p(\pmb\psi,\pmb\gamma\mid\pmb Y,\pmb\ell)\,
    p(\pmb\ell\mid\pmb D),
\end{equation}
where the sum is over allocations of the unknown labels: \(\pmb\ell \in \{1,\ldots,K\}^{\otimes n}\), \(p(\pmb\ell\mid\pmb D)\) is their marginal posterior probability, and \(p(\pmb\psi,\pmb\gamma\mid\pmb Y,\pmb\ell)\) is the complete-label downstream posterior. 
%
Using Bayes' theorem to express \(p(\pmb\psi,\pmb\gamma\mid\pmb Y,\pmb\ell)\) in terms of \(p(\pmb Y\mid\pmb\psi,\pmb\gamma,\pmb\ell)\), 
\eqref{eq:full-posterior-decomposition} can be rewritten as 
\begin{multline*}
p(\pmb\psi,\pmb\gamma\mid\pmb D) = \sum_{\pmb\ell}
    \frac{p(\pmb Y\mid \pmb\psi,\pmb\gamma,\pmb\ell) p(\pmb\psi,\pmb\gamma\mid\pmb\ell)}{p(\pmb Y \mid \pmb\ell)}
    \, p(\pmb\ell\mid\pmb D)
    \\ = \sum_{\pmb\ell} \left[
    \prod_{i} p_{l_i}(\pmb Y_i\mid\pmb\psi, \pmb\gamma_{l_i})
    \right] \frac{p(\pmb\psi,\pmb\gamma\mid\pmb\ell)}{p(\pmb Y \mid \pmb\ell)}
    \, p(\pmb\ell\mid\pmb D).
\end{multline*}

\begin{figure}[h!]\centering
    \resizebox{0.9\linewidth}{!}{%
        \includegraphics{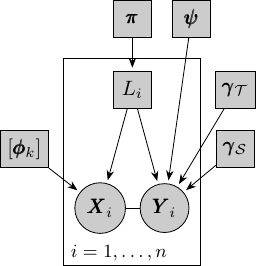}%
        \hspace{1em}%
        \includegraphics{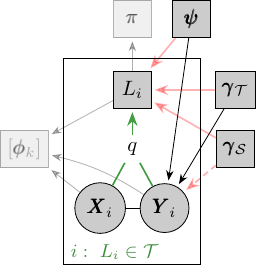}%
    }%
    \caption{\it Graphical representations (where circles denote observed data and squares unknown parameters) of the full joint model (left) and our re-formulation for modular inference with potentially mis-specified mixture components (right). The latter results from inverting the edges from \(L_i\) to \(\pmb X_i\) and \(\pmb Y_i\) and performing two modifications: restricting the influence (red solid edges) of \(\pmb\psi\) and \(\pmb\gamma\) on \(L_i\), sampling it instead from a fixed distribution \(q\) conditioned on the data; and removing suspect components \(\mathcal S\) both from the model (red dashed edge \(\pmb\gamma_{\mathcal S} \to \pmb Y_i\)) and the data (keeping only those objects for which the sampled \(L_i\) is a trusted component). The faded parameters \([\pmb\phi_k]\) and \(\pmb\pi\) are deemed not of scientific interest but can be easily inferred, conditional on sampled \(\pmb L_i\), similarly to \(\pmb\psi\), \(\pmb\gamma\).\label{fig:graph}}
\end{figure}

We focus on settings in which \(p_k(\pmb Y_i\mid\pmb\psi, \pmb\gamma_k)\) is not trusted for one or more components. Let
\(\mathcal S\subset\{1,\ldots,K\}\) denote these suspect components and
\(\mathcal T=\{1,\ldots,K\}\setminus\mathcal S\) the trusted components.
The downstream likelihoods \(p_k\) of the suspect components may be misspecified,
while those of the trusted components are considered reliable for
inference
on \(\pmb\psi\) and the trusted-component-specific parameters \(\pmb\gamma_k\).

Misspecification can affect the analysis via two distinct paths. First, the conditional class probabilities  
\begin{multline}\label{eq:label-post-conditional}
    P(L_i=k\mid\pmb X_i,\pmb Y_i,\pmb\theta)
    \\ =
    \frac{
        \pi_k
        p_k(\pmb Y_i\mid\pmb\psi,\pmb\gamma_k)
        h_k(\pmb X_i\mid\pmb Y_i,\pmb\phi_k)
    }{
        \sum_j
        \pi_j
        p_j(\pmb Y_i\mid\pmb\psi,\pmb\gamma_j)
        h_j(\pmb X_i\mid\pmb Y_i,\pmb\phi_j)
    }
\end{multline}
are affected by \emph{all} component likelihoods through the denominator.
Consequently, misspecification of a suspect component \(j\in\mathcal S\) can distort the conditional allocation probabilities for all \(k\), and this distortion propagates to the posterior allocation distribution
\begin{equation}\label{eq:p(L|D)}
    p(\pmb L=\pmb\ell\mid\pmb D) =
    \int p(\pmb L=\pmb\ell\mid\pmb D,\pmb\theta) \, p(\pmb\theta\mid\pmb D)\,\dd\pmb\theta.
\end{equation}
Moreover, the generative model \(h_k(\pmb X\mid\pmb Y,\pmb\phi_k)\), required for making use of \(X\), may be difficult to specify reliably. Finally, the suspect-component likelihoods contribute directly to inference on \(\pmb\psi\) and can thus distort it even if the class labels were perfectly well known.

\subsection{Modular construction}

To prevent misspecification in the suspect components from propagating into inference on the parameters of interest, we introduce two modular interventions. The first replaces the full-model allocation probabilities by fixed external probabilities, e.g., learned directly from labelled data, preventing the (possibly mis-specified) downstream model from updating them. The second removes the suspect-component likelihoods from inference on the shared and trusted component parameters. The resulting modular construction is shown in the right panel of Fig.~\ref{fig:graph}. 

\paragraph{Fixed allocation model.} 
Because \(p(\pmb L\mid\pmb D,\pmb\theta)\) generally depends on \(\pmb\theta\), which is inferred with the full model, potentially misspecified downstream component likelihoods can influence the inferred labels.

To motivate an alternative, consider the ideal setting in which the joint model is correctly specified. Under standard posterior consistency
conditions, \(p(\pmb\theta\mid\pmb D)\) concentrates at the true parameter
\(\pmb\theta_0\) as the sample size increases, so that the marginal allocation posterior from \eqref{eq:p(L|D)} approaches 
\begin{equation}\label{eq:label-post-factorisation}
p(\pmb L\mid\pmb D,\pmb\theta=\pmb\theta_0)
=
\prod_i p(L_i\mid\pmb D_i,\pmb\theta=\pmb\theta_0). 
\end{equation}
The conditional class probabilities on the right can be estimated directly from representative labelled observations \((\pmb X_j,\pmb Y_j,L_j)\) from the true data-generating process \(P_0\) using a flexible possibly high-dimensional probabilistic classifier, without relying on an explicit mixture model. This bypasses possible mis-specifications of \(p_k(\pmb Y\mid\pmb\psi,\pmb\gamma_k)\) and/or
\(h_k(\pmb X\mid\pmb Y,\pmb\phi_k)\) while exploiting information in both \(X\) and \(Y\).
Alternatively, since our construction is agnostic to the provenance of the fixed probabilistic allocations, unsupervised probabilistic clustering can be used when labelled data are not available.

While this argument motivates a classifier-based replacement, it is stronger than required for reliable inference. Our theoretical results below show that we do not need to recover \(P_0(L\mid\pmb X,\pmb Y)\) pointwise. A weaker downstream-conditional calibration condition is sufficient to preserve the target of the corresponding true-label analysis.

Now, let \(q_k(\pmb X_i,\pmb Y_i)\), \(k=1,\ldots,K\), denote the predicted probability that observation \(i\) belongs to component \(k\).
For observations with known labels, we set
\(q_k(\pmb X_i,\pmb Y_i)=\mathbf 1\{L_i=k\}\).
In accordance with \eqref{eq:label-post-factorisation}, the classifier-based allocation probability is
\begin{equation}
    q(\pmb\ell\mid\pmb D)
    =\prod_{i=1}^n q_{\ell_i}(\pmb X_i,\pmb Y_i).
\end{equation}
Our first intervention consists in swapping \(p(\pmb\ell\mid\pmb D)\) for \(q(\pmb\ell\mid\pmb D)\) in \eqref{eq:full-posterior-decomposition}, defining the \emph{classifier-based posterior}
\begin{equation}
    \widetilde p_q(\pmb\psi,\pmb\gamma\mid\pmb D)
    =
    \sum_{\pmb\ell}
    p(\pmb\psi,\pmb\gamma\mid\pmb Y,\pmb\ell)\,
    q(\pmb\ell\mid\pmb D).
\end{equation}

Note that \(q(\pmb\ell\mid\pmb D)\), like \(p(\pmb\ell\mid\pmb D)\), is fixed during downstream inference (e.g., Monte Carlo posterior sampling) of \( (\pmb\psi,\pmb\gamma) \), while full joint inference requires \(p(\pmb\ell\mid\pmb D, \pmb\theta)\) in \eqref{eq:label-post-conditional} to be continuously updated. Fixing \(q\) thus also prevents feedback from the allocation probabilities to the downstream inference.

\paragraph{Protected downstream inference.}
The second intervention further protects the shared and trusted-component parameters from misspecification in the suspect components. Let
\(\pmb\gamma_{\mathcal T}
=\{\pmb\gamma_k:k\in\mathcal T\}\).
For a fixed allocation \(\pmb\ell\), we define the protected complete-label posterior by deliberately omitting the suspect component likelihoods 
\begin{equation}
    p_{\mathcal T}
    (\pmb\psi,\pmb\gamma_{\mathcal T}\mid\pmb Y,\pmb\ell)
    \propto
    p(\pmb\psi,\pmb\gamma_{\mathcal T})
    \prod_{i:\ell_i\in\mathcal T}
    p_{\ell_i}
    (\pmb Y_i\mid\pmb\psi,\pmb\gamma_{\ell_i}).
\end{equation} 
Conditional on the labels, this is the ordinary posterior based on the
observations allocated to the trusted components. If the trusted-component
likelihoods are correctly specified, 
Bayesian inference for \((\pmb\psi,\pmb\gamma_{\mathcal T})\) under \(p_{\mathcal T}\) is consistent for \((\pmb\psi_0,\pmb\gamma_{\mathcal T,0})\) under the usual regularity conditions, irrespective of misspecification of the suspect components. 
If the suspect components are also correctly specified, the corresponding full and protected posteriors have the same asymptotic target, although the protected analysis may be less efficient. 
This intervention thus trades potential constraining power for reliability.

\paragraph{Combined modular posterior and computation.}
Combining the two interventions gives the protected modular posterior
\begin{align}
    \widetilde p_{q,\mathcal T}
    (\pmb\psi,\pmb\gamma_{\mathcal T}\mid\pmb D)
    &=
    \sum_{\pmb\ell}
    p_{\mathcal T}
    (\pmb\psi,\pmb\gamma_{\mathcal T}\mid\pmb Y,\pmb\ell)\,
    q(\pmb\ell\mid\pmb D).
    \label{eq:modular-posterior}
\end{align}
Thus, uncertainty in the unknown labels is propagated by averaging the protected complete-label posteriors with respect to the fixed allocation distribution. The sum, which in general has \(k^n\) terms of varying significance (according to \(q(\pmb\ell\mid\pmb D)\)), can be approximated directly by Monte Carlo:
\begin{align}\label{eq:iykyk}
    \widetilde p_{q,\mathcal T}
    (\pmb\psi,\pmb\gamma_{\mathcal T}\mid\pmb D)
    & \approx
    \frac{1}{M}
    \sum_{m=1}^M
    p_{\mathcal T}
    (\pmb\psi,\pmb\gamma_{\mathcal T}
     \mid\pmb Y,\pmb L^{(m)})
    \\\nonumber
    \text{with}\quad
    \pmb L^{(m)} &\sim q(\pmb L\mid\pmb D)
    \quad\text{for } m=1,\ldots,M.
\end{align}

Algorithm~\ref{algo:procedure} summarizes the practical workflow this implies: sample label allocations from \(q\), fit the protected downstream model for each allocation, and average the corresponding posteriors. Importantly, each allocation receives equal weight in \eqref{eq:iykyk}, which means that the protected complete-label posteriors need to be properly normalized. If the individual posteriors are represented through Monte Carlo samples, this corresponds to combining an equal number of (effectively independent) draws from each.

\begin{algorithm}[t]
\caption{Computing the modular posterior}
\begin{algorithmic}[1]
\REQUIRE Data \(\pmb D=(\pmb X,\pmb Y)\), fixed allocation probabilities
\(q(\pmb L\mid\pmb D)\), trusted components \(\mathcal T\), number of
allocations \(M\).
\FOR{\(m=1,\ldots,M\)}
    \STATE Allocate labels via
    \(\pmb L^{(m)}\sim q(\pmb L\mid\pmb D)\). 
    \STATE Form trusted subset
    \(\pmb Y_{\mathcal T}^{(m)}
    =\{\pmb Y_i:L_i^{(m)}\in\mathcal T\}\).
    \STATE 
    Obtain  
    \(
    p_{\mathcal T}
    (\pmb\psi,\pmb\gamma_{\mathcal T}
    \mid\pmb Y^{(m)}_{\mathcal T},\pmb L^{(m)}).
    \)
\ENDFOR
\STATE Average the \(M\) posteriors with equal weight \(1/M\). 
\end{algorithmic}
\label{algo:procedure}
\end{algorithm}

\subsection{Theoretical guarantees}\label{sec:theory}

%
Let \(P_0\) denote the joint distribution of the data and labels under the true data-generating model, and let \(\mathbb E_0\) denote expectation with respect to \(P_0\). 
\begin{theorem}[Consistency under conditional calibration]
Let \(q=(q_1,\ldots,q_K)\) be a fixed probabilistic classifier satisfying
\begin{align}\label{eq:calibration_condition}
    \mathbb E_0
    \!\left[q_k(\pmb X,\pmb Y)\mid\pmb Y\right]
    =
    P_0(L=k\mid\pmb Y)
\end{align}
for \(k=1,\ldots,K\). We refer to \eqref{eq:calibration_condition} as conditional calibration with respect to \(\pmb Y\). 
Let 
\(L_i^\star\mid\pmb X_i,\pmb Y_i
\sim q(\cdot\mid\pmb X_i,\pmb Y_i)\)
independently across \(i\), 
and suppose the complete-label posterior
\(p(\pmb\psi,\pmb\gamma\mid\pmb Y,\pmb L)\) is consistent for a fixed target 
\((\pmb\psi^\dagger, \pmb\gamma^\dagger)\).
Then 
\[
    (\pmb Y,\pmb L^\star)\overset d=(\pmb Y,\pmb L),
\]
and the classifier-based posterior
\(\widetilde p_q(\pmb\psi,\pmb\gamma\mid\pmb D)\) is consistent for the same target \((\pmb\psi^\dagger, \pmb\gamma^\dagger)\).

Moreover, whenever posterior means are integrable,
\begin{align}
\mathbb E_0\!\left[
    \mathbb E_{\widetilde p_q}\!\left\{(\pmb\psi,\pmb\gamma)\mid\pmb D\right\}
\right]
=
\mathbb E_0\!\left[
    \mathbb E\!\left\{(\pmb\psi,\pmb\gamma)\mid\pmb Y,\pmb L\right\}
\right],
\label{eq:thm1_expectation}
\end{align} 
i.e., the classifier-based posterior mean has the same expectation under \(P_0\) as the complete-label posterior mean that would be obtained if the true labels were observed, 
and therefore the same finite-sample bias   
relative to \((\pmb\psi^\dagger,\pmb\gamma^\dagger)\).
\label{thm:conditional-calibration}
\end{theorem}

A proof is given in \ref{sec:proof}. The theorem provides a direct link between classifier calibration and valid downstream inference. Under conditional calibration with respect to \(\pmb Y\), stochastic labels drawn from \(q\) preserve the joint distribution of the downstream data and true labels, \((\pmb Y,\pmb L^\star)\overset d=(\pmb Y,\pmb L)\).
Consequently, the classifier-based posterior shares the asymptotic target of the corresponding true-label analysis and has the same finite-sample expected posterior mean, and hence the same bias relative to any fixed target.
The theorem concerns the classifier-based posterior obtained by averaging over the full allocation distribution \(q(\pmb L\mid\pmb D)\). In practice, this average is approximated using \(M\) sampled allocations. 

Combining Theorem~\ref{thm:conditional-calibration} with the protected
complete-label analysis gives the corresponding guarantee for the full
modular procedure. If the trusted-component likelihoods are correctly
specified and satisfy the usual conditions for posterior consistency, and
\(q\) satisfies \eqref{eq:calibration_condition}, then
\(\widetilde p_{q,\mathcal T}\) is consistent for
\((\pmb\psi_0,\pmb\gamma_{\mathcal T,0})\).
Moreover, whenever the corresponding posterior means are integrable,
\eqref{eq:thm1_expectation} also applies to the protected analysis, so that
the classifier-based protected posterior mean has the same expectation
under \(P_0\) as the corresponding complete-label protected posterior mean,
and therefore the same finite-sample bias. 

\paragraph{Implications for classifier choice and validation.}
Both oracle classifiers
\(q^{XY}_{0,k}(\pmb X,\pmb Y) =P_0(L=k\mid\pmb X,\pmb Y)\) and \(q^{Y}_{0,k}(\pmb Y)=P_0(L=k\mid\pmb Y)\)
satisfy \eqref{eq:calibration_condition}. For \(q^{XY}_0\), this follows from iterated expectations, 
while for \(q^Y_0\) it is immediate. In contrast,
\(q^{X}_{0,k}(\pmb X)=P_0(L=k\mid\pmb X)\) need not satisfy
\eqref{eq:calibration_condition} and therefore need not preserve the joint distribution of \(\pmb Y\) and the latent class.

Interestingly, Theorem~\ref{thm:conditional-calibration} shows that pointwise recovery of the true class probabilities is not required. Classifier errors or scatter are allowed, provided that they average out conditional on \(\pmb Y\), as in \eqref{eq:calibration_condition}. This yields a practical validation criterion when representative labelled data (held out or cross-fitted when those data are used for training) are available. After evaluating \(q_k(\pmb X_i,\pmb Y_i)\) on them, one can assess whether the residuals
\(
    \mathbf 1\{L_i=k\}
    -
    q_k(\pmb X_i,\pmb Y_i)
\)
have mean zero conditional on \(\pmb Y_i\), for example using bins or
smooth estimates of the conditional mean. 
The classifier may also be trained using external data or simulations, provided that the fitted \(q\) satisfies \eqref{eq:calibration_condition} under the target data-generating distribution.
Agreement with this condition provides direct evidence that the classifier is suitable for the downstream
analysis.

\begin{figure*}[t]
    \centering
    \includegraphics[width=0.99\textwidth]{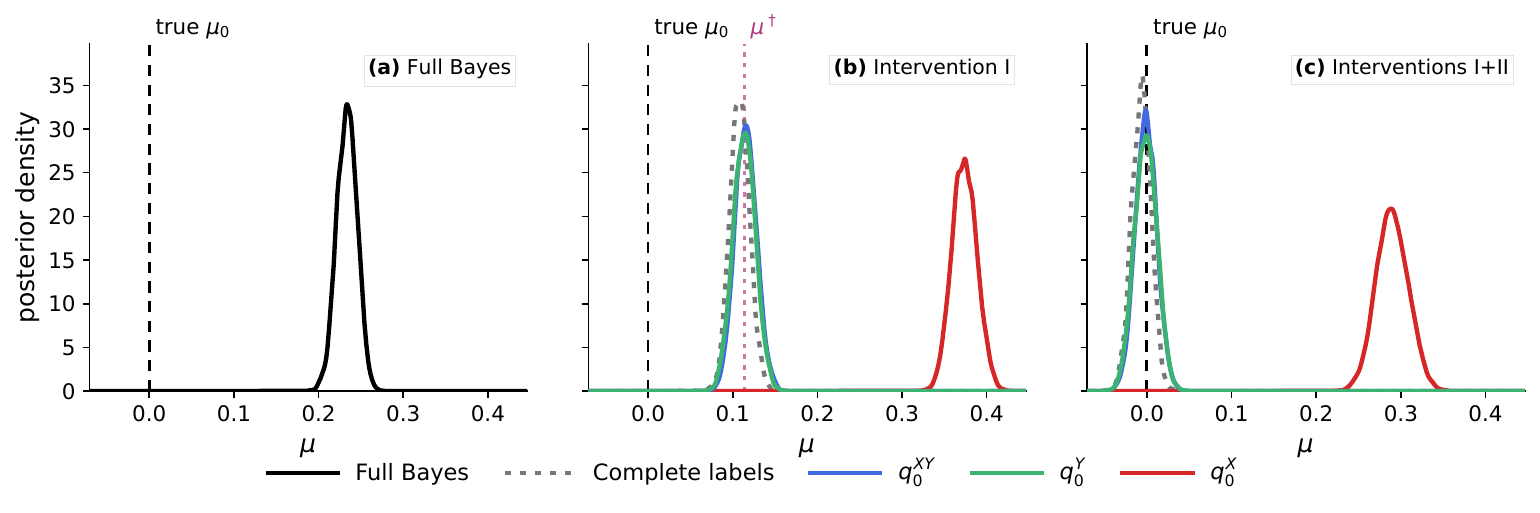}
    \caption{\it Posterior densities for the shared mean \(\mu\). (a) Full Bayes, (b) Intervention I with fixed classifier probabilities, and (c) Interventions I+II with the suspect component excluded. The dashed black line marks the true \(\mu_0\); the magenta dotted line marks \(\mu^\dagger\), the asymptotic target of the complete-label posterior under the full downstream model.\label{fig:shared_mean_post}}
\end{figure*}

\section{Numerical illustration}\label{sec:toy_example}

\paragraph{Simulation setup.} We consider a simple three-component mixture to illustrate how localized
component misspecification can distort inference on a shared parameter of
interest and how the proposed modular construction protects against this
contamination. We generate the labels and auxiliary data according to 
\[
L_i\sim\operatorname{Categorical}(\pmb\pi),
\qquad
X_i\mid L_i=k\sim N(a_k,\tau_k^2). 
\] 
The downstream quantity of interest is a mean parameter \(\mu\) shared
across all three components. Components \(1\) and \(3\) are trusted and generated as 
\[
Y_i\mid L_i=k
\sim N(\mu_0,\sigma_k^2),
\qquad k\in\{1,3\}.
\]
Component \(2\) is treated as suspect. We generate
\[
Y_i
=
\mu_0+\sigma_2 Z_i+\delta E_i,
\qquad
Z_i\sim N(0,1),\quad
E_i\sim\operatorname{Exp}(1),
\]
independently, so that \(\delta\) controls the degree of misspecification.
For inference, we assume the Gaussian downstream model 
\(Y_i\mid L_i=k,\mu,\sigma_k\sim N(\mu,\sigma_k^2)\),
\(k=1,2,3\).

For this illustration, we set
\(n=5{,}000\),
\(\pmb\pi=(0.3,0.4,0.3)\),
\(\pmb a=(-0.6,0,0.8)\),
\(\tau_k=0.5\), \(\mu_0=0\),
\((\sigma_1,\sigma_2,\sigma_3)=(0.8,0.4,0.5)\), 
and \(\delta=1\). Ten percent of the class labels are revealed and treated as known.

\paragraph{Posterior targets under misspecification.}
We first fit the fully joint Bayesian mixture model described in Section~\ref{sec:notation_and_fullyBayes}, in which the latent labels and downstream parameters are inferred jointly. 
Panel (a) of Fig.~\ref{fig:shared_mean_post} shows that the resulting posterior for \(\mu\) is substantially shifted away from the true value \(\mu_0\),  showing how misspecification in component \(2\) can distort inference on the shared parameter.

Panel (b) applies Intervention I by replacing the model-based label posterior with the fixed oracle classification probabilities \(q_0^X\), \(q_0^Y\), and \(q_0^{XY}\). 
For comparison, we also include the complete-label benchmark, assuming all labels are observed, to remove classification uncertainty. Relative to the fully Bayesian analysis, the complete-label posterior is substantially less shifted from \(\mu_0\), but still concentrates at the misspecified target \(\mu^\dagger\neq\mu_0\), derived analytically in the Supplement.
Consistently with the result in Section~\ref{sec:theory}, \(q_0^Y\) and \(q_0^{XY}\), which satisfy \eqref{eq:calibration_condition}, recover the same complete-label target, whereas \(q_0^X\) does not.

Panel (c) combines Interventions I and II by additionally excluding the suspect component from the downstream analysis. 
As shown in Section~\ref{sec:theory}, because the retained components are correctly specified, the protected complete-label posterior targets \(\mu_0\), and the protected analyses based on \(q_0^Y\) and \(q_0^{XY}\) recover the same target, whereas \(q_0^X\) remains shifted.

\paragraph{Classification properties and efficiency.} Fig.~\ref{fig:shared_mean_cond_calibration} provides an empirical check of \eqref{eq:calibration_condition}. At a given value of \(Y\), the smoothed mean classifier probability estimates the left-hand side of \eqref{eq:calibration_condition}, while the empirical class fraction among labelled observations estimates the right-hand side.
These closely agree for \(q_0^Y\) and \(q_0^{XY}\), whereas \(q_0^X\) shows a clear discrepancy, consistent with the posterior distortion in Fig.~\ref{fig:shared_mean_post}. 
Among observations with similar \(Y\), \(q_0^{XY}\) also exhibits greater variation than \(q_0^Y\), reflecting the additional label information carried by \(X\). (The small local spread of \(q_0^Y\) arises from smoothing over nearby values of \(Y\).) Correspondingly, the posterior based on \(q_0^{XY}\) is slightly more concentrated.

In the Supplement, we further illustrate that exact recovery of \(q_0^{XY}\) is unnecessary by considering perturbed probabilities \(q_i^{(\kappa)}\sim\operatorname{Dirichlet}(\kappa q_{0,i}^{XY})\), where smaller \(\kappa\) produces stronger perturbations. Since the perturbations are centered on \(q_0^{XY}\), their errors average out (conditionally on \(Y\)), and even strong perturbations leave the downstream posterior nearly unchanged.

Lastly, this example highlights that standard classification metrics need not reflect the desired properties of a classifier to be used for downstream posterior inference. For example, \(q_0^X\) has higher accuracy and macro one-vs-rest AUC than \(q_0^Y\) (\(0.664\) versus \(0.550\), and \(0.838\) versus \(0.732\)), yet \(q_0^X\) yields a shifted posterior while \(q_0^Y\) recovers the complete-label target.  
Generally, conditional calibration (i.e., correct average class probabilities at fixed \(Y\)) preserves the complete-label \emph{target}, while individual predictive power (e.g., from auxiliary information that resolves label uncertainty at fixed \(Y\)) can improve posterior \emph{efficiency}.

\begin{figure}
    \centering    \includegraphics[width=0.9\linewidth]{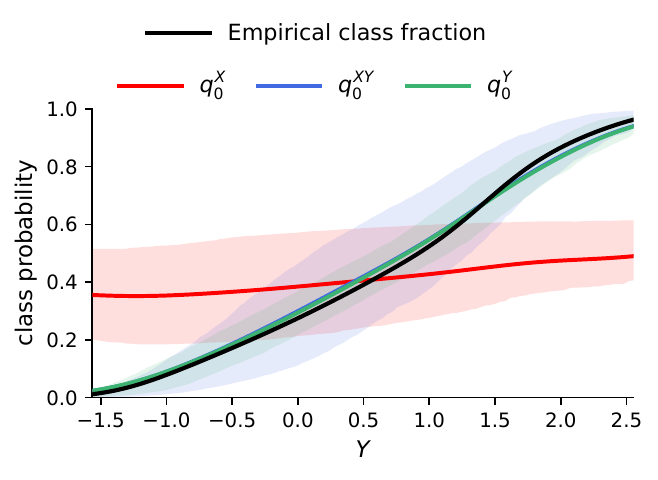}
    \caption{\it Conditional-calibration diagnostic for class 2 in the numerical illustration. The black curve shows the empirical class fraction among labelled observations; coloured curves and bands show the local means and 25th--75th percentile ranges of \(q_0^X\), \(q_0^{XY}\), and \(q_0^Y\). The empirical class fraction is estimated from the labelled observations, while the classifier means and percentile ranges are computed from all observations using Gaussian-kernel smoothing in \(Y\). Agreement of the means with the empirical fraction indicates conditional calibration with respect to \(Y\). \label{fig:shared_mean_cond_calibration}}
\end{figure}

\section{SN Ia case study}

\begin{figure*}[t!]\centering
    \includegraphics[width=0.9\linewidth, trim=1.24cm 0 0 0]{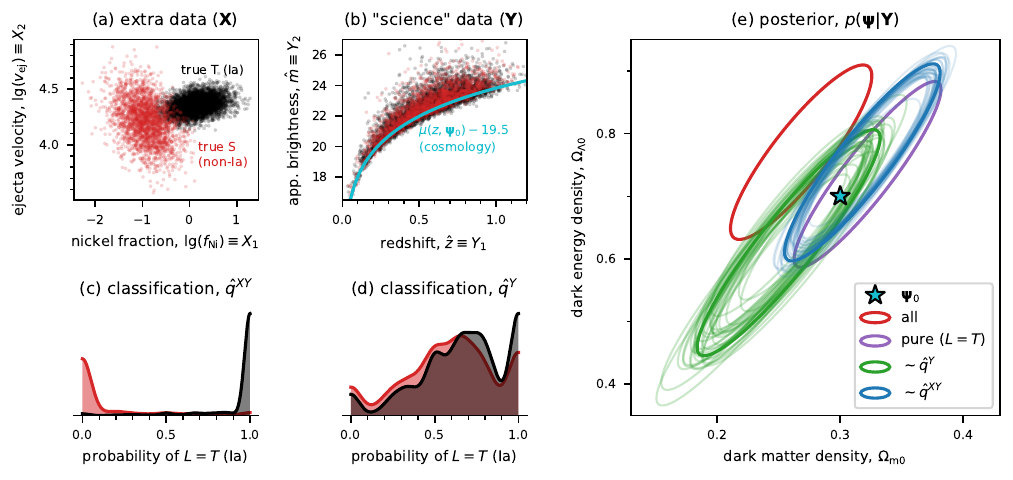}%
    \caption{\it SN Ia cosmology in the presence of contamination. (a) and (b): the data, where classification is informed mainly by \(\pmb X\). This is reflected in the distribution of classifier probabilities derived with \(\hat{q}^{XY}\) and \(\hat{q}^{Y}\), shown respectively in (c) and (d). Finally, (e): posteriors (\(2\sigma\), i.e.\@ \(86\,\%\) credible regions of the Gaussian approximation to the MCMC posterior) derived from the full data, 
    the pure trusted subset (using the true labels (\(L=T\)) unavailable in practice),  
    and our proposed procedure using \(\hat{q}^{Y}\) and \(\hat{q}^{XY}\), as labelled. Note that both of the latter recover the input cosmology, but \(\hat{q}^{Y}\) introduces 
    subsampling noise 
    (the individual \(p_{\mathcal{T}}(\pmb\psi\mid\pmb Y^{(m)}_{\mathcal{T}})\) are shown as pale contours), whereas the result with \(\hat{q}^{XY}\) more closely resembles the posterior from pure data.\label{fig:sn}}
\end{figure*}

In this section, we present a cosmological analysis of mock SN data modelled on the Legacy Survey of Space and Time (LSST), currently under way at the Vera Rubin Observatory, combined with spectroscopic follow-up with 4MOST/TiDES \citep{Frohmaier_2025}, which will obtain robust typing (\(L\)) for \(\sim 1\,\%\) and accurate host redshifts (\(\hat{z}\)) for \(\sim 10\,\%\) of the observed transients.

\paragraph{Mock data.} Focusing in this work only on the deep-drilling field (DDF) and inspired by the predictions of \citet{Mitra_2025}, we generate \({\approx}5000\) (mock) SNe Ia, distributed according to the volumetric rate from PLAsTiCC \citep{Kessler_2019}. We simulate them using the Bayesian hierarchical model (BHM) of \citet{mandel2017type}, which we will later use for the analysis; it is hence well-specified and trusted. As data, in addition to the precisely observed redshift (\(\hat{z}=z\)), we consider the results of fitting the SALT light-curve model\footnote{For simplicity, we do not perform fits (or simulate light curves); instead, we adopt diagonal uncertainties of the sampling distributions as in \citet{boyd2026flowsn}, based on the latent parameters in the BHM (see \ref{sec:simple-bayesn}). As the uncertainties are a-priori unknown, they are technically part of the data, but since they will not be informative of classification, we omit them from the discussion.}, which gives three summary statistics per object, representing peak brightness (\(\hat{m}\)), stretch (\(\hat{x}\)), and colour (\(\hat{c}\)); that is \(\pmb Y_i = (\hat{z}_i, \hat{m}_i, \hat{x}_i, \hat{c}_i)\). As additional observables \(\pmb X_i\) that help with classification but are not used in the cosmological analysis, we take the two explosion parameters (derived from photometric light curves) discussed by \citet{Uzsoy_2026}; we simulate them from the distribution(s) in the top-right panel of Fig.~8 therein. 

Various transient classes \citep[see e.g.][]{Kessler_2019,Frohmaier_2025} are expected to contaminate SN Ia surveys. Here, we will only include one suspect population matched to that of the SNe Ia except for a drift of brightness versus redshift, \(\Delta M = 0.2 \times z\), which can bias the inferred cosmology.\footnote{In reality, different contaminants will have different volumetric (redshift) rates, and their light curves fitted with SALT will lead to various data distributions. We adopt the present model to better illustrate our methodological arguments in this scientific context.} In practice, the true data-generating process of non-Ia SNe is not reliably known, motivating us to treat it as a suspect component, meaning that $K=2$ and, using for clarity $T\equiv1$ and $S\equiv2$, $\mathcal{T}=\{T\}$ and $\mathcal{S}=\{S\}$. We set \(\pi_S = 1/3\) contamination\footnote{This is in accordance with estimated non-Ia rates \citep{Frohmaier_2025}, but we note that the quality cuts typically applied to the fitted light curves can significantly reduce this fraction to about \(10\,\%\) even without explicit classification.} and construct a labelled training set of \({\approx}3000\) SNe Ia and \({\approx}1500\) contaminants, sampled from the true populations.

\paragraph{Classifier.} We train a simple 4-dimensional binned classifier in \(\pmb Y\)-space via  
Voronoi tessellation of a random subset of \(1000\) training examples.  
We set \(\hat q^Y\) to the true-Ia fraction among all training examples in
each cell, and analogously \(\hat q^{XY}\) in 6-dimensional
\((\pmb X,\pmb Y)\)-space. 
These classifiers are constructed to approximately satisfy the calibration condition \eqref{eq:calibration_condition}. 

\paragraph{Analysis.} Using each classifier, we generate \(M=30\) realizations of the class allocations and construct the corresponding trusted subsets \(\pmb Y^{(m)}_{\mathcal{T}}\). To analyze each of them, we employ the probabilistic model \(\prod_i p_T(\pmb Y_i\mid\pmb\psi,\pmb\gamma_T, L_i=T)\) elaborated fully in \ref{sec:simple-bayesn}, which was also used for the mock generation. It describes the SN Ia population using hyperparameters for the distributions of stretch, colour, and absolute brightnesses (\(M_i\)), as well as coefficients for the correlations between them. A cosmological model for the distance modulus (\(\mu\)) converts absolute to apparent brightnesses \(m_i = M_i + \mu(z_i, \pmb\psi)\), depending on redshift and the cosmological parameters \(\pmb\psi = (\Omega_{\rm m0}, \Omega_{\Lambda 0})\), the present-day relative dark-matter and dark-energy densities, which are the primary target of scientific interest. We use affine-invariant MCMC sampling of the global parameters \(\pmb\psi, \pmb\gamma_T\) via the \texttt{emcee} package \citep{emcee} and the marginal likelihood given by \citet{mandel2017type}, using 32 walkers and taking 1500 draws after 500 burn-in steps. 
We obtain a Monte Carlo sample from the approximate 
posterior distribution in  \eqref{eq:iykyk} by combining an equal number of draws from each MCMC run.

\paragraph{Results.} We present the results in Fig.~\ref{fig:sn}. As anticipated, we recover the input cosmology using both \(\hat{q}^{Y}\) and \(\hat{q}^{XY}\). Significant scatter is present among the individual posteriors \(p_{\mathcal{T}}(\pmb\psi\mid\pmb Y^{(m)}_{\mathcal{T}})\) in the former case, leading to a slightly less constraining (by \({\sim}20\,\%\)) final result. Using the extra data \(\pmb X\) to improve pointwise classification naturally moves the constraints towards those from an analysis of the pure subset. These results show that full Bayesian sampling of class probabilities is not needed for reliable and powerful cosmological inference from SN surveys and open the door to sophisticated SBI applications in which only the well-understood Ia population is modelled, while contaminants are probabilistically discarded at inference.

\section{Discussion}

The proposed modularization trades some potential efficiency for robustness to
localized misspecification. This trade-off is particularly attractive when
suspect populations are scientifically secondary and contribute relatively
little information about the parameters of interest, while their misspecification
can distort downstream inference. If suspect-component parameters are also of
interest, the protected posterior can subsequently be propagated into the
suspect module without reverse feedback, as in standard cut inference.

A practical implication is that classifier quality should be evaluated with the
downstream analysis in mind rather than by predictive performance alone. Our
derived calibration condition provides such a criterion, which can be assessed
when representative labelled data are available. Future work should develop
practical diagnostics for assessing calibration and the efficiency gains from
incorporating additional classification information, address covariate shift when
labelled data are non-representative \citep{autenrieth2024stratified}, and study
validation of probabilistic allocations when labels are unavailable. Our approach
also simplifies computation, which may facilitate its use across a range of
scientific applications and inference methodologies (e.g.\@, simulation-based inference).

    
\subsubsection*{Acknowledgements}
MA and KSM are supported by the European Union’s Horizon 2020 research and innovation programme under European Research Council Grant Agreement No. 101002652 (BayeSN; PI K. Mandel).   
DvD and MA acknowledge partial support from the UK Engineering and Physical Sciences Research Council [EP/W015080/1, EP/W522673/1].


\bibliography{bibliography}  

\newpage

\setcounter{section}{0}
\renewcommand\thesection{S\arabic{section}}

\setcounter{figure}{0}
\renewcommand\thefigure{S\arabic{figure}}

\setcounter{table}{0}
\renewcommand\thetable{S\arabic{table}}

\renewcommand{\theequation}{S.\arabic{equation}}
\renewcommand{\theHequation}{S\arabic{equation}}

\onecolumn

\begin{center}
{\Large\bfseries Supplementary Materials}
\end{center}
\vspace{1em}

\section{Additional theoretical results and proofs}

\subsection{Proof of Theorem~\ref{thm:conditional-calibration}}\label{sec:proof}

In this section, we present the detailed proof of Theorem~\ref{thm:conditional-calibration}

\begin{proof}[Proof of Theorem~\ref{thm:conditional-calibration}]
For each \(i\), let
\(L_i^\star\mid X_i,Y_i\sim q(\cdot\mid X_i,Y_i)\) independently, and
write \(\mathbb E_q[\cdot\mid\pmb D]\) for expectation over the imputed
labels \(\pmb L^\star\) conditional on \(\pmb D\).

By construction,
\(P(L_i^\star=k \mid X_i,Y_i)=q_k(X_i,Y_i)\). Hence, by  
iterated expectations 
and the assumed conditional calibration of \(q\) in \eqref{eq:calibration_condition},
\[
    P(L_i^\star=k\mid Y_i)
    =
    \mathbb E_0[q_k(X_i,Y_i)\mid Y_i]
    =
    P_0(L_i=k\mid Y_i).
\]
Thus \(L_i^\star\mid Y_i\) and \(L_i\mid Y_i\) have the same
distribution. 
Hence \((Y_i,L_i^\star)\overset d=(Y_i,L_i)\), 
and by independence across $i$,
\begin{align}
(\pmb Y,\pmb L^\star)\overset d=(\pmb Y,\pmb L).
    \label{eq:YL-distributional-identity}
\end{align}
To establish consistency of
\(\widetilde p_q(\pmb\psi,\pmb\gamma\mid\pmb D)\), let
\(\epsilon>0\) and define
\[
A_\epsilon
=
\left\{
(\pmb\psi,\pmb\gamma):
\left\|(\pmb\psi,\pmb\gamma)
-(\pmb\psi^\dagger,\pmb\gamma^\dagger)\right\|>\epsilon
\right\},
\qquad
R_n(\pmb Y,\pmb L)
=
p(A_\epsilon\mid\pmb Y,\pmb L).
\]
By consistency of the complete-label posterior,
\(R_n(\pmb Y,\pmb L)\xrightarrow{P_0}0\), and because \(0\leq R_n\leq1\), 
\[
    \mathbb E_0[R_n(\pmb Y,\pmb L)]\longrightarrow0.
\]

By \eqref{eq:YL-distributional-identity} and the definition of the
classifier-based posterior,
\[
\mathbb E_0\!\left[
    \widetilde p_q(A_\epsilon\mid\pmb D)
\right]
=
\mathbb E_0\!\left[
    \mathbb E_q\{R_n(\pmb Y,\pmb L^\star)\mid\pmb D\}
\right]
=
\mathbb E_0[R_n(\pmb Y,\pmb L)]
\longrightarrow 0.
\]
For any \(\eta>0\), Markov's inequality then gives
\[
P_0\!\left\{
    \widetilde p_q(A_\epsilon\mid\pmb D)>\eta
\right\}
\leq
\frac{
    \mathbb E_0[\widetilde p_q(A_\epsilon\mid\pmb D)]
}{\eta}
\longrightarrow0.
\]
Thus
\(\widetilde p_q(A_\epsilon\mid\pmb D)\xrightarrow{P_0}0\)
for every \(\epsilon>0\), proving consistency for
\((\pmb\psi^\dagger,\pmb\gamma^\dagger)\). Finally, \eqref{eq:thm1_expectation} follows directly from \eqref{eq:YL-distributional-identity} by iterated expectation, assuming integrability.
\end{proof}

We note that the i.i.d.\ assumption is made for simplicity; the same argument extends
to independent, non-identically distributed observations when the
calibration condition holds observation-wise.

\section{Additional details and results for Section~\ref{sec:toy_example}}\label{sec:supp_toy_example}

\paragraph{Complete-label misspecification target \(\mu^\dagger\).}
The complete-label posterior concentrates, under standard misspecified-Bayes
regularity conditions, at the parameter values for which the downstream model
is closest, in Kullback--Leibler divergence and averaged over the true class
proportions, to the true class-conditional generative distributions of \(Y\).
Equivalently,
\[
(\mu^\dagger,\pmb\gamma^\dagger)
=
\arg\min_{\mu,\pmb\gamma}
\sum_{k=1}^K
\pi_k\,
\mathrm{KL}\!\left[
P_0(Y\mid L=k)
\,\Vert\,
p_k(\,\cdot\mid\mu,\gamma_k)
\right].
\]

For the Gaussian downstream model in Section~\ref{sec:toy_example},
profiling over the component variances gives
\[
\sigma_k^{2\dagger}(\mu)
=
\operatorname{Var}_0(Y\mid L=k)
+
\{E_0(Y\mid L=k)-\mu\}^2.
\]
Substituting the generative distributions from
Section~\ref{sec:toy_example} into the resulting profiled
Kullback--Leibler objective and minimizing with respect to \(\mu\)
gives, 
\(
\mu^\dagger=0.1139
\), as illustrated by the vertical magenta dotted line in Fig.~\ref{fig:shared_mean_post}. 

\paragraph{Oracle classification probabilities.}
For the numerical illustration, we compute the oracle classifiers from the
true generative distributions. Writing \(f_k^X\) and \(f_k^Y\) for the true
class-conditional densities of \(X\) and \(Y\), respectively,
\[
q_{0,k}^X(X_i)
=
\frac{\pi_k f_k^X(X_i)}
{\sum_{l=1}^K \pi_l f_l^X(X_i)},
\qquad
q_{0,k}^Y(Y_i)
=
\frac{\pi_k f_k^Y(Y_i)}
{\sum_{l=1}^K \pi_l f_l^Y(Y_i)}.
\]
In the simulation, \(X\) and \(Y\) are conditionally independent given \(L\),
so that
\[
q_{0,k}^{XY}(X_i,Y_i)
=
\frac{\pi_k f_k^X(X_i)f_k^Y(Y_i)}
{\sum_{l=1}^K \pi_l f_l^X(X_i)f_l^Y(Y_i)}.
\]
Here, \(f_k^X\) is the Gaussian distribution used to generate \(X\) in each
component. For \(Y\), \(f_k^Y\) is Gaussian for the trusted components and is
the true shifted Gaussian--exponential convolution for the suspect component. 
For observations whose labels are known, the corresponding classification probabilities are replaced by the appropriate one-hot vector in the downstream analyses.

\clearpage

\subsection{Additional illustrations for Section~\ref{sec:toy_example}.}

We provide additional results supporting the numerical illustration in
Section~\ref{sec:toy_example}, including the full class-wise
conditional-calibration diagnostic and further classifier comparisons.

Figure~\ref{fig:supp_cond_calibration} extends the conditional-calibration
diagnostic in Figure~\ref{fig:shared_mean_cond_calibration} to all three
classes. In each case, the mean probabilities under \(q_0^Y\) and
\(q_0^{XY}\) closely follow the empirical class fraction, whereas \(q_0^X\)
does not. \(q_0^{XY}\) shows greater variation at fixed \(Y\), reflecting
the additional information carried by \(X\). 

\begin{figure*}[h]\centering
    \includegraphics[width=0.99\textwidth]{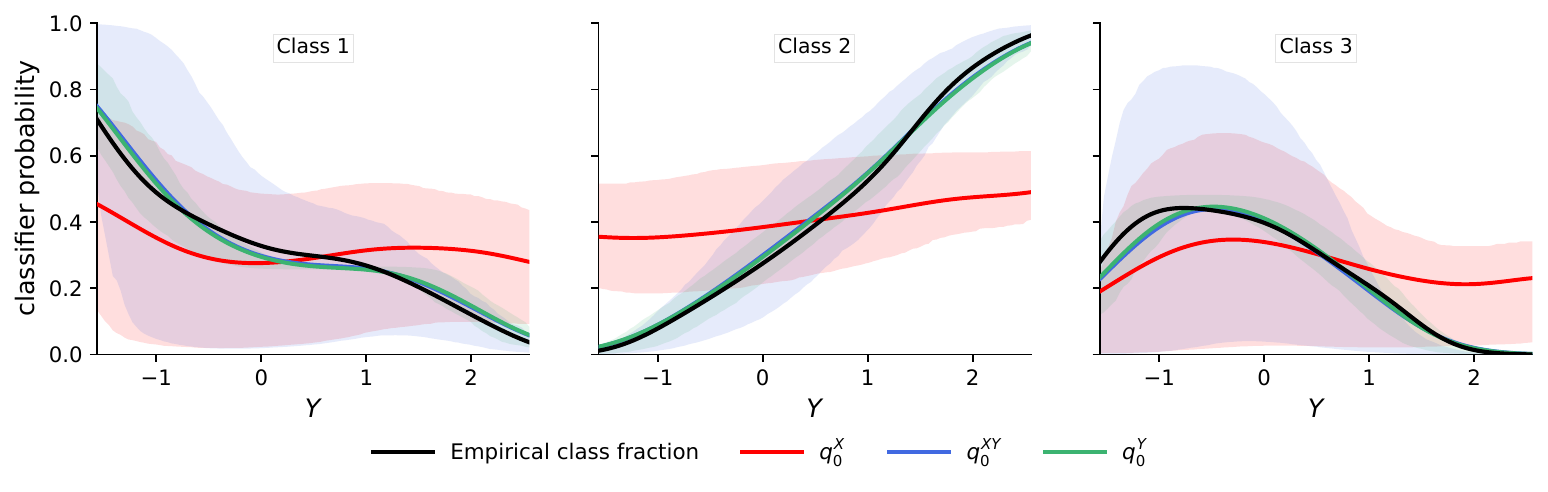}
    \caption{Conditional-calibration diagnostic for all three classes in the numerical illustration. Black curves show the empirical class fractions among labelled observations; coloured curves and shaded regions show the local means and 25th--75th percentile ranges of \(q_0^X\), \(q_0^{XY}\), and \(q_0^Y\). The class-2 panel is shown in Figure~\ref{fig:shared_mean_cond_calibration} in the main text.
\label{fig:supp_cond_calibration}}
\end{figure*}

Figure~\ref{fig:noisy_vs_qxy} illustrates the Dirichlet perturbations of
\(q_0^{XY}\) described in Section~\ref{sec:toy_example}. The perturbations can
substantially alter individual class probabilities, particularly for
\(\kappa=5\), while remaining centered on the oracle probabilities; increasing
\(\kappa\) reduces this pointwise variation. 
Nevertheless, Figure~\ref{fig:noisy_vs_oracle_mu} shows that these perturbations only mildly affect inference on \(\mu\). This illustrates Theorem~\ref{thm:conditional-calibration}, which shows that pointwise recovery of the oracle class probabilities is not required provided classification errors average out conditional on \(Y\), as in \eqref{eq:calibration_condition}.

\begin{figure*}\centering
    \includegraphics[width=0.99\textwidth]{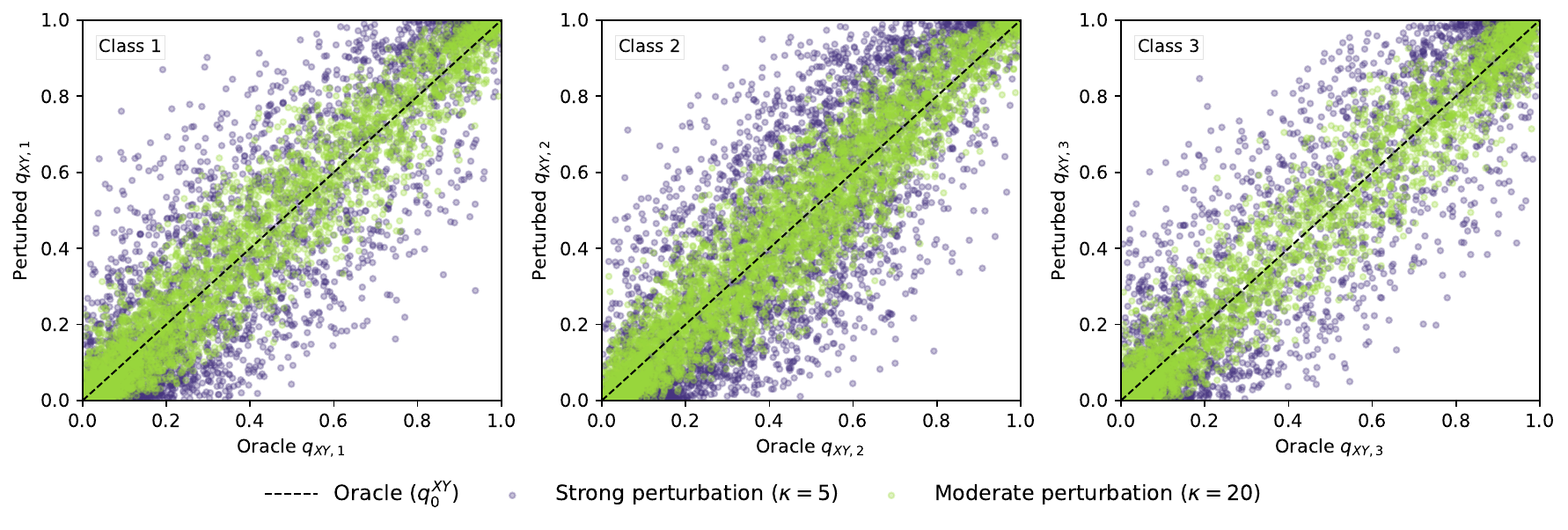}
    \caption{Dirichlet perturbations of the oracle classifier \(q_0^{XY}\) for
each class in the numerical illustration. Perturbed class probabilities are
plotted against their oracle values; the dashed line denotes equality.  \label{fig:noisy_vs_qxy}}
\end{figure*}

\begin{figure*}\centering
\includegraphics[width=0.5\textwidth]{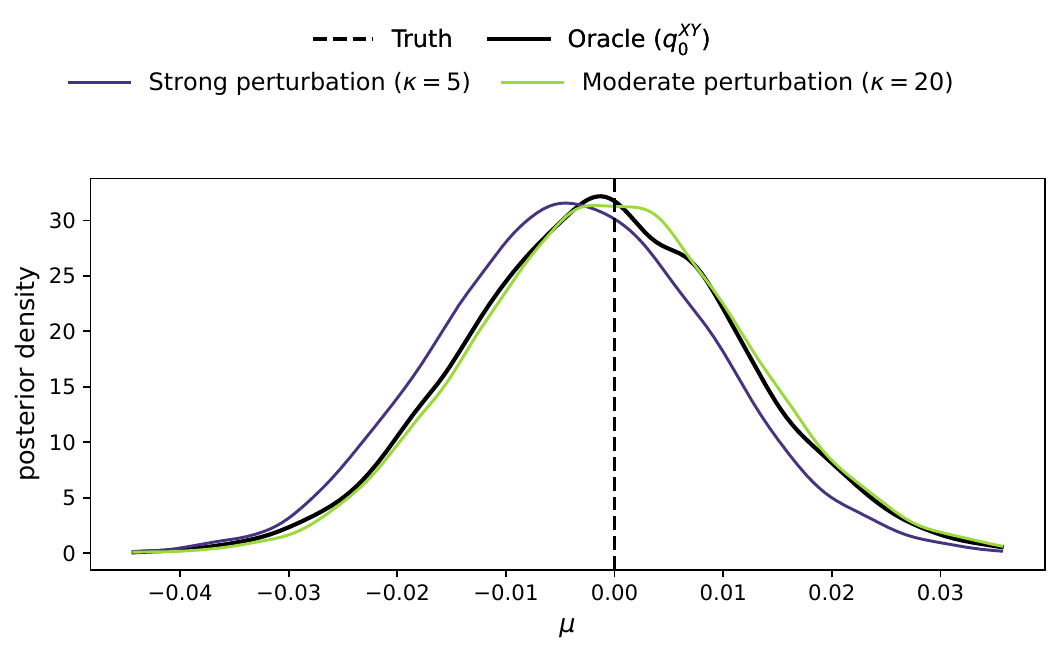}
    \caption{Posterior densities for \(\mu\) under the oracle \(q_0^{XY}\) and
its Dirichlet perturbations. Despite substantial pointwise perturbations of
the class probabilities, the resulting posteriors change only mildly. The
dashed line marks the true \(\mu_0\).  \label{fig:noisy_vs_oracle_mu}}
\end{figure*}

Finally, Figure~\ref{fig:AUC} further illustrates that standard
classification performance need not reflect downstream inferential behaviour.
Although \(q_0^X\) has a higher macro one-vs-rest AUC than \(q_0^Y\)
(\(0.838\) versus \(0.732\)), it leads to the substantially shifted posterior
for \(\mu\) in Section~\ref{sec:toy_example}, whereas \(q_0^Y\) satisfies
\eqref{eq:calibration_condition} and recovers the complete-label target.
Conversely, the perturbed \(q_0^{XY}\) classifiers have lower AUC than the
oracle classifier while producing very similar downstream posteriors. 

\begin{figure*}\centering
    \includegraphics[width=0.99\textwidth]{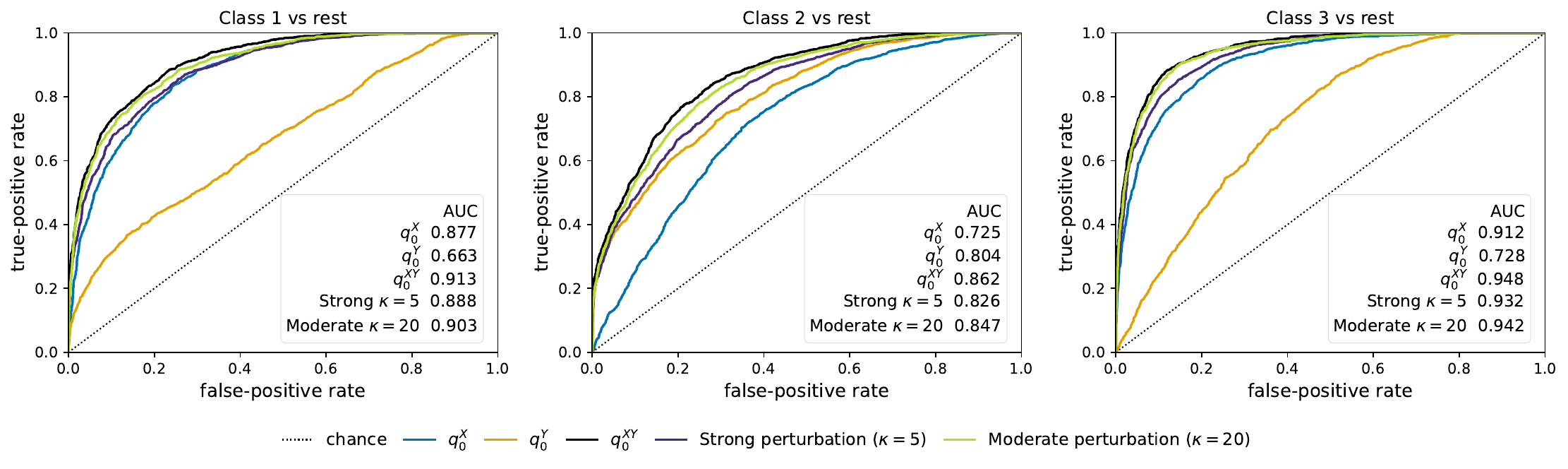}
    \caption{One-vs-rest ROC curves for the oracle and perturbed classifiers in the numerical illustration. AUC values are reported within each panel.  \label{fig:AUC}}
\end{figure*}

\clearpage

\section{Simple Bayesian hierarchical model for SN Ia cosmology}\label{sec:simple-bayesn}

Even though both the seminal cosmological analyses of type Ia supernovae \citep{riess1998observational,perlmutter1999measurements} and the current flagship efforts \citep{Brout_2022a,Popovic_2025a} rely on \(\chi^2\) fitting to infer cosmology, it has been widely recognised \citep[e.g.][]{march2011improved} that this can lead to biases due to the hierarchical structure of the model, in which observational uncertainties are intertwined with population scatter and poorly constrained physical effects (e.g.\@, dust extinction). Bayesian hierarchical modelling represents the principled approach to the problem, but it is yet to be applied to a full-scale cosmological analysis, starting from photometric light curves. A simpler approach is to model analytically distributions of a small number of summary statistics, e.g. as derived by fitting a parametrised light-curve template.

One such model, whose centrepiece is the physically motivated treatment of dust extinction, was developed by \citet{mandel2017type} and can be summarised (in the forward factorisation) as follows. For each SN Ia, labelled by \(i\), first sample latent intrinsic stretch, (\(B-V\))colour, and absolute brightness from a correlated multivariate normal distribution:
\begin{align}
    x_{\rm int}^i & \sim \mathcal{N}(\mu_x, \sigma_x^2), \\
    c_{\rm int}^i & \sim \mathcal{N}(\mu_c + \alpha_c x_{\rm int}^i, \sigma_c^2), \\
    M_{\rm int}^i & \sim \mathcal{N}(M_0 + \alpha x_{\rm int}^i + \beta c_{\rm int}^i, \sigma_0^2),
\intertext{where \(\mu_x, \sigma_x, \mu_c, \sigma_c, M_0, \sigma_0, \alpha, \alpha_c, \beta\) are global parameters. Then sample dust parameters; in this work, we assume a single shared total-to-selective extinction \(R_B\) and colour excess sampled from an exponential prior with rate \(1/\tau\) (a global parameter):}
    E_{B-V}^i & \sim\mathcal{E}(1/\tau);
\intertext{and apply (deterministically) extinction to the brightness and attenuation to the colour parameters, leaving the stretch unchanged:}
    x_{\rm ext}^i & = x_{\rm int}^i, \\
    c_{\rm ext}^i & = c_{\rm int}^i + E_{B-V}, \\
    M_{\rm ext}^i & = M_{\rm int}^i + E_{B-V} R_B.
\intertext{Finally, convert absolute to (noiseless) apparent brightness via the (cosmological) distance modulus:}
    m^i & = M_{\rm ext} + \mu(z^i, \pmb\psi),
\end{align}
where \(z^i\) is the SN's cosmological redshift, and \(\pmb\psi\) are the parameters of the cosmological model, for which we choose the standard (non-flat) \(\Lambda\)--cold dark matter with parameters \(\Omega_{\Lambda0}, \Omega_{\rm m0}\) describing the relative densities of its two cosmic components. Additionally, the Hubble constant (\(H_0\)) defines the absolute distance scale, but since it is fully degenerate with the absolute brightness scale (\(M_0\)), we choose the arbitrary value \(H_0 = 100\,\mathrm{km}\,\mathrm{s}^{-1}\,\mathrm{Mpc}^{-1}\).

To complete the model, we need to specify a sampling distribution that connects the latent parameters to observables. We assume that fit uncertainties follow the distributions presented in \citet{boyd2026flowsn}, derived from template fits to realistic light curves:
\begin{align}
    \ln \Sigma_x^i & \sim \mathcal{N}(-3, 1^2), \\
    \ln \Sigma_c^i & \sim \mathcal{N}(-7, 0.6^2), \\
    \ln \Sigma_m^i & \sim \mathcal{N}(0.2(m^i-56), 1.2^2).
\end{align}
Finally, we use those to draw noisy summary statistics \(\hat{x}^i \sim \mathcal{N}(x_{\rm ext}^i, \Sigma_x^i)\), \(\hat{c}^i \sim \mathcal{N}(c_{\rm ext}^i, \Sigma_c^i)\), \(\hat{m}^i \sim \mathcal{N}(m^i, \Sigma_m^i)\), regarded as observed, together with the redshift, for which we assume infinite precision: \(\hat{z}^i = z\).

In this work, for simplicity, we fix the population distributions to \(\mu_x = 0\), \(\sigma_x = 1\), \(\mu_c=-0.08\), \(\sigma_c=0.03\), \(\alpha_c=0\), \(\tau=0.15\), and \(R_B=4.1\), assigning broad uniform priors to the remaining \(\Omega_{\rm m0}, \Omega_{\Lambda 0}, M_0, \sigma_0, \alpha, \beta\), which we sample using MCMC and the marginal likelihood presented in \citet{mandel2017type}.


\end{document}